\documentclass[journal,10pt]{IEEEtran}

\usepackage[T1]{fontenc}
\usepackage[utf8]{inputenc}
\usepackage{textcomp}
\usepackage{microtype}
\usepackage{amsmath,amssymb,amsfonts,amsthm}
\usepackage{booktabs}
\usepackage{graphicx}
\usepackage{xcolor}
\usepackage{tikz}
\usepackage{algorithm}
\usepackage{algpseudocode}
\usepackage[numbers,sort&compress]{natbib}
\definecolor{linkpurple}{RGB}{92,35,145}
\usepackage[colorlinks=true,linkcolor=linkpurple,citecolor=linkpurple,urlcolor=linkpurple]{hyperref}
\usetikzlibrary{arrows.meta,positioning,shapes.geometric,fit,backgrounds}

\newcommand{\Description}[1]{}
\theoremstyle{definition}
\newtheorem{theorem}{Theorem}
\newtheorem{lemma}{Lemma}
\newtheorem{proposition}{Proposition}

\newcommand{\fallbackfigure}[2]{%
	\fbox{\parbox[c][0.22\textheight][c]{#1}{\centering #2}}
}

\newcommand{\optionalfigure}[3]{%
	\IfFileExists{#1}{\includegraphics[width=#2]{#1}}{\fallbackfigure{#2}{#3}}
}

\title{\LARGE In-Orbit Intelligence or Ground Offloading?\\[0.15ex]
Inference Freshness under Intermittent Satellite Connectivity}
\author{Ayşe~Nur~Pehlivanoğlu, Aimin~Li, and Elif~Uysal,~Fellow,~IEEE%
\thanks{The authors are with the Communication Networks Research Group (CNG), METU, Ankara, Türkiye (e-mail: ayse.pehlivanoglu@metu.edu.tr; aimin@metu.edu.tr; uelif@metu.edu.tr).}}

\begin{document}

\maketitle

	\begin{abstract}
		%This paper studies inference freshness in a hybrid satellite architecture with onboard and ground inference under intermittent connectivity.At each decision epoch, the satellite chooses among waiting, onboard computation, cached semantic transmission, and raw-data offloading. We adopt Age of Inference (\textsc{AoInf}) as the performance metric, where the age resets only upon successful task-valid updates. We formulate long-run average \textsc{AoInf} minimization as a finite-state average-cost semi-Markov decision process whose state captures the ground \textsc{AoInf}, orbital contact phase, cache occupancy, and cache age. We then transform the SMDP into an equivalent average-cost MDP and compute the policy via normalized relative value iteration (RVI). Numerical results show that the hybrid policy reduces average \textsc{AoInf} relative to random, onboard-only, and offload-only baselines.
		This paper studies how to balance onboard and ground computation under intermittent LEO connectivity for optimized inference freshness. As connectivity varies in time, the system switches among the actions of onboard computation, cached semantic transmission, raw-data offloading, and waiting. We define Age of Inference (\textsc{AoInf}) as the performance metric, where the age resets only upon successful task-valid updates. We formulate long-run average \textsc{AoInf} minimization as a finite-state average-cost semi-Markov decision process whose state captures the ground \textsc{AoInf}, orbital contact phase, cache occupancy, and cache age. We then transform the SMDP into an equivalent average-cost MDP and compute the solution via normalized relative value iteration (RVI). Numerical results indicate that the resulting hybrid policy reduces average \textsc{AoInf} relative to  onboard-only and offload-only baselines, while requiring less computational resources on the satellite than the former, and fewer communication resources than the latter.
	\end{abstract}

	
	\section{Introduction}
	\subsection{Background}
	LEO satellite networks are emerging as an orbital edge layer for sensing, communication, and time-sensitive inference. With
	onboard processors and higher-rate satellite-ground links, the bottleneck is shifting from raw-observation downlink to the timely delivery of task-relevant inference before losing operational value~\cite{George2018OnboardProcessing,Bhattacherjee2020InOrbitComputing,Wang2023SatelliteComputing,Yin2025OrbitalEdgeSurvey}. In disaster response, maritime monitoring, and IoT sensing, operators often require goal-oriented outcomes such as alerts, target presence, or anomaly labels rather than full data reconstruction, in line with recent discussions of goal-oriented communications for non-terrestrial networks~\cite{Gunduz2023BeyondBits,uysal2024goalorient,Li2026FreshnessEffectiveness}.
	
	\begin{figure*}[t]
		\centering
		\includegraphics[width=0.7\textwidth]{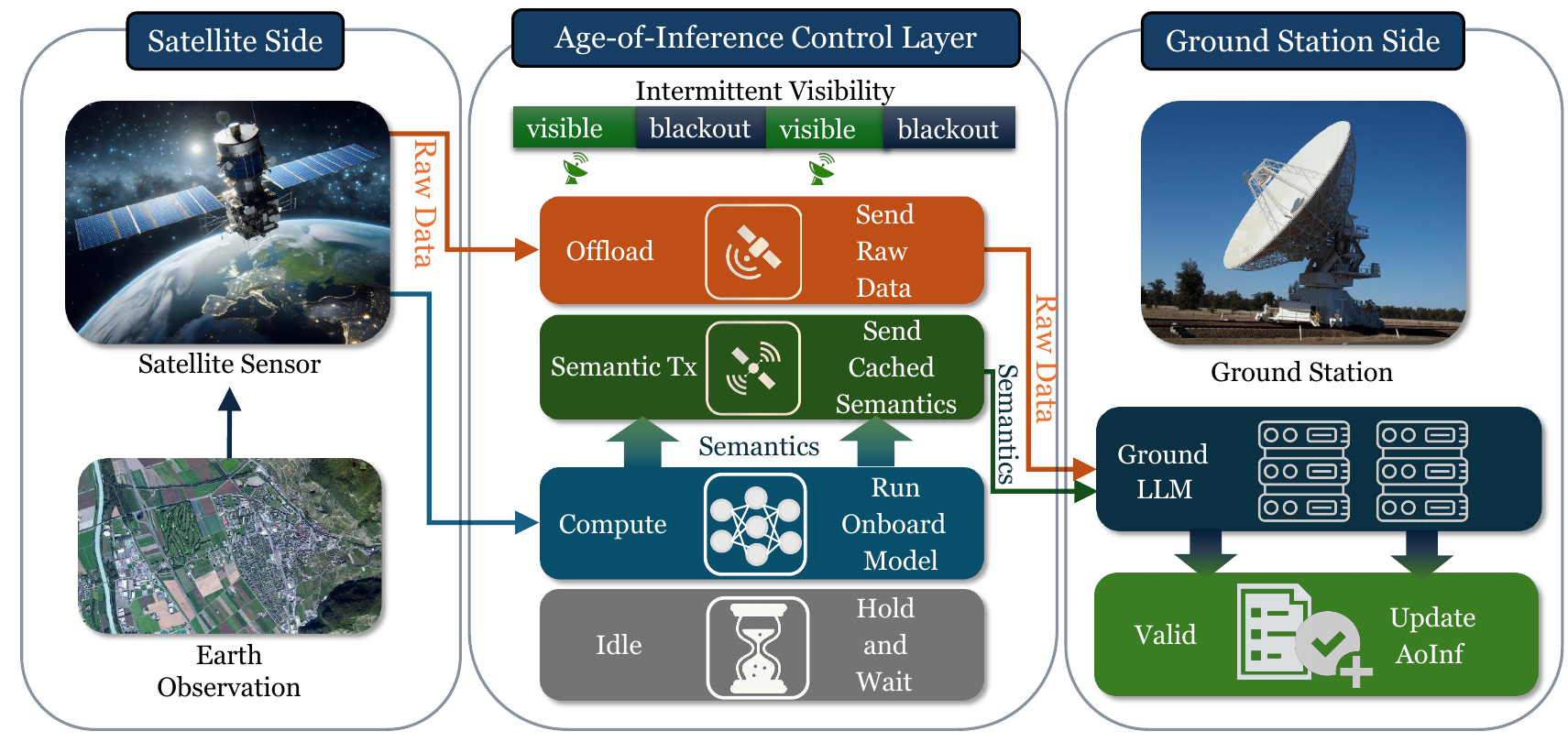}
		\caption{Contact-constrained hybrid inference system. The controller observes the ground AoInf, contact phase, cache
			occupancy, and cache age, and then chooses among waiting, onboard recomputation, cached semantic transmission, and
			raw-data offloading.}
		\Description{A system diagram showing the logic of the hybrid controller. Sensor observations are processed onboard into semantic results, stored in a cache, and passed to a controller that observes the Age of Inference, contact phase, cache occupancy, and cache age. The controller may idle, recompute, transmit cached semantics, or offload raw data to the ground. Ground inference and validation update the Age of Inference only on success.}
		\label{fig:system-logic}
	\end{figure*}
	
	Such outcomes arise in a Satellite-Terrestrial Integrated Network (STIN) through onboard semantic computing, semantic transmission, or raw-data offloading. Although STIN architectures distribute computation across orbital and terrestrial infrastructure~\cite{Xie2020SatelliteTerrestrialEdge,Tang2021LEOHybridCloudEdge,Wang2020GameSatelliteEdge,Zhai2024FedLEO}, intermittent visibility and residual contact time still govern the best path. A short contact window may carry a compact \textit{semantic summary}\footnote{We use \emph{semantic summary} to denote an onboard, goal-oriented representation, such as a confidence-scored label, a detected region, or a compressed feature embedding.}, but not a raw data payload.
	
	These options make inference freshness path-dependent. The conventional Age of Information (AoI) metric measures the age of the latest received update~\cite{Yates2021AoISurvey}, while correctness-aware and goal-oriented metrics ask whether a timely update is useful for the task~\cite{Maatouk2020AoII,QueryAoI,Maatouk2023AoIIEnabler,litensor1,Gan2025TaskOrientedAoI}. Remote-inference studies further show that freshness, representation choice, and inference accuracy interact in nontrivial ways~\cite{Shisher2022FreshnessSupervisedLearning,Shisher2023LearningCommunicationCodesign,Shisher2024TimelyRemoteInference,Li2026FreshnessEffectiveness}. We therefore use Age of Inference (\textsc{AoInf}), which resets only when a delivered update yields a task-valid inference outcome. Because cached semantics and offloaded observations have different generation and delivery delays, semantic transmission, raw-data offloading, and onboard computation create different reset opportunities rather than interchangeable packet deliveries.
	
	\textbf{Motivation}. The satellite must balance computation, storage, and transmission under intermittent satellite-ground links. Our guiding question is: \textit{under intermittent satellite connectivity, when should a satellite compute, transmit semantics, offload raw data, or wait to keep ground-side inference fresh?}
	
	\subsection{Related Works and Contributions}
	Existing work gives parts of this picture, but not the complete decision loop. Satellite edge-computing and satellite-terrestrial integrated network studies show how computation and communication can be distributed across orbital nodes, ground stations, and terrestrial clouds \cite{Soret2021SatelliteIoTOffloadBackhaul, Li2022AgeDrivenSatelliteIoT, Zarini2025LEOBDRISAoI, Delfani2025SemanticsLEOSatellites}. Age-oriented satellite updating models capture the role of timely delivery under intermittent
	access, but typically omit heterogeneous inference paths and
	semantic cache reuse \cite{Jiao2023AgeOptimalSatelliteIoT,soret20215g,11161975}. Remote-inference and goal-oriented communication studies show that freshness, representation, and accuracy should be jointly optimized \cite{ Maatouk2023AoIIEnabler, Li2024GoalOrientedSemantic, Shisher2024TimelyRemoteInference, Gan2025TaskOrientedAoI,uysal2024goalorient}. However, the optimal choice among computation, caching, offloading, and idling depends on contact dynamics, cache state, and action-specific latencies.
	
	We develop a contact-constrained LEO hybrid inference model capturing contact windows, onboard computation, semantic caching, and raw-data offloading. The main contributions are as follows:
	\begin{enumerate}
		\setlength{\topsep}{2pt}
		\setlength{\partopsep}{0pt}
		\setlength{\parsep}{0pt}
		\setlength{\itemsep}{2pt}
		\item \textbf{Model and metric.} We model hybrid orbital inference as a finite-state average-cost SMDP over contact phase, \textsc{AoInf}, cache occupancy, and cache age. Unlike throughput- or latency-oriented formulations~\cite{Xie2020SatelliteTerrestrialEdge,Bhattacherjee2020InOrbitComputing,Tang2021LEOHybridCloudEdge}, we optimize inference freshness, with \textsc{AoInf} resetting only after a task-valid update. Unlike AoI-oriented satellite updating and resource-allocation models that omit heterogeneous inference paths or semantic caching~\cite{Soret2021SatelliteIoTOffloadBackhaul,Li2022AgeDrivenSatelliteIoT,Jiao2023AgeOptimalSatelliteIoT,Zarini2025LEOBDRISAoI,Delfani2025SemanticsLEOSatellites}, our model captures \textit{computation, semantic transmission, offloading, and waiting}\footnote{Waiting can be optimal for minimizing AoI under certain conditions~\cite{Sun2017UpdateOrWait}; we therefore include it as a control action.} under intermittent satellite connectivity.
			
		\item \textbf{Solution methodology and threshold structure.} We transform the SMDP into an equivalent average-cost MDP and solve it via normalized RVI. We establish monotonicity and threshold-type switching conditions among actions.
			
		\item \textbf{Numerical evidence.} We show that the optimized hybrid policy adapts to cache freshness, residual visibility, and inference-success reliability, reducing average \textsc{AoInf} relative to random, onboard-only, and offload-only baselines.
	\end{enumerate}
	
	\section{System Model}
	\subsection{Overall System View}
	We consider a slotted-time hybrid inference system with a sensing satellite, an onboard semantic processor with a one-item cache, and a ground inference destination, as shown in Fig.~\ref{fig:system-logic}. At each decision epoch, the controller observes the ground \textsc{AoInf}, contact phase, and cache state, then chooses to wait, compute a fresh semantic summary, transmit cached semantics, or offload a fresh raw observation. The objective is to keep ground-side inference timely and task-valid despite intermittent connectivity.
	
	\subsection{Intermittent Satellite Connectivity}
	Within this overall architecture, communication opportunities are periodic. The satellite-ground link follows a period of length $P$, within which the link is available for exactly $W \leq P$ consecutive slots and unavailable for the remaining slots. Let $\phi_t \in \{0,1,\dots,P-1\}$ denote the contact phase at decision epoch $t$. The remaining visible time is
	\begin{equation}
		r(\phi_t) \triangleq
		\begin{cases}
			W-\phi_t, & \phi_t \in \{0,1,\dots,W-1\},\\
			0, & \phi_t \in \{W,\dots,P-1\}.
		\end{cases}
	\end{equation}
	The quantity \(r(\phi_t)\) determines communication feasibility. Semantic transmission must finish within \(r(\phi_t)\). For offloading, only the raw upload must fit in the visible window; ground inference may continue after contact ends.
	
	\subsection{Age of Inference}
	Let \(g_j\) be the generation time of update \(j\), \(d_j\) the slot when update \(j\) becomes available at the ground, and \(Z_j \in \{0,1\}\) indicate whether update \(j\) yields a successful inference outcome. The ground Age of Inference (\textsc{AoInf}) at slot \(n\) is:
	\begin{equation}
		\Delta_G(n) \triangleq n-\nu(n)
	\end{equation}
	where $\nu(n) \triangleq \max\{g_j : d_j \le n,\; Z_j=1\}$. If no successful update has arrived by slot \(n\), we set \(\Delta_G(n)=\hat{\Delta}\). Thus, \textsc{AoInf} is refreshed only by updates that are delivered and inference-successful. Between successful updates, \(\Delta_G(n)\) increases linearly with unit slope until it reaches the cap \(\hat{\Delta}\).
	
	\subsection{Objective and Decision Epochs}
	The optimal long-run time-average \textsc{AoInf} at the ground is
	\begin{equation}\label{rhostardefinition}
		\rho^\star \triangleq\min_{\pi} \limsup_{N\to\infty}\frac{1}{N}\,
		\mathbb{E}_\pi\!\left[\sum_{n=1}^{N}\Delta_G(n)\right].
	\end{equation}
	This criterion favors fresh task-relevant inference rather than raw-data delivery. Decisions are
	made only when the current action completes; because the next
	state and elapsed time both depend on the action, the problem is
	an SMDP~\cite{puterman1994mdp}.
	
	\subsection{SMDP Formulation}
	\subsubsection{State Space.}
	At each epoch $t$, the system state is:
	\begin{equation}
		s_t \triangleq (\Delta_t,\phi_t,q_t,\tau_t),
	\end{equation}
	where $\Delta_t \in \{1,2,\dots,\hat{\Delta}\}$ denotes the ground \textsc{AoInf} at the beginning of the epoch $t$, $\phi_t$ is the contact phase, $q_t \in \{0,1\}$ is the cache-occupancy indicator with $q_t=1$ when the semantic cache contains a semantic summary and $q_t=0$ when the cache is empty, and $\tau_t \in \{0,1,\dots,\hat{\Delta}\}$ denotes the age of the cached semantic summary, with $\tau_t=0$ whenever $q_t=0$. We collect the non-\textsc{AoInf} variables in the mode descriptor
	\begin{equation}
		m_t = (\phi_t,q_t,\tau_t),
	\end{equation}
	so that the full state is $(\Delta_t,m_t)$. The admissible state space:
	\begin{equation}
		\mathcal{S}
		=
		\Bigl\{
		(\Delta,m):
		q=0\Rightarrow \tau=0
		\Bigr\}.
		\label{eq:state-space}
	\end{equation}
	
	\subsubsection{Action Space.}
	At each decision epoch, the controller selects one action from the candidate action set
	\begin{equation}
		\mathcal{A}_0=\{\textit{idle},\textit{compute},\textit{tx},\textit{offload}\}.
	\end{equation}
	The admissible subset depends on the current mode \(m=(\phi,q,\tau)\). The four candidate actions are defined as follows:
	{\setlength{\parindent}{0pt}
		\begin{list}{}{%
				\setlength{\leftmargin}{0pt}%
				\setlength{\labelwidth}{0pt}%
				\setlength{\labelsep}{0pt}%
				\setlength{\itemindent}{0pt}%
				\setlength{\listparindent}{0pt}%
				\setlength{\topsep}{2pt}%
				\setlength{\partopsep}{0pt}%
				\setlength{\parsep}{0pt}%
				\setlength{\itemsep}{2pt}%
			}
			\item \textbf{1) Idle.} Under \textit{idle}, the controller waits for one slot without a new inference update with action duration \(L_I=1\).
			
			\item \textbf{2) Compute.} Under \textit{compute}, onboard inference generates a fresh semantic summary and overwrites the cache. The action duration is \(L_C=C_S\). The observation is generated at the start of computation, so the cached summary is \(L_C\) slots old.
			
			\item \textbf{3) Transmit cached semantics.} Under \textit{tx}, the controller transmits the currently cached semantic summary to the ground. The action duration is \(L_T=U_{\mathrm{tx}}\). This action is feasible only if the cache is nonempty and the semantic summary can be delivered before the current contact window closes, i.e., \(q=1\) and \(r(\phi)\ge U_{\mathrm{tx}}\).
			
			\item \textbf{4) Offload raw data.} Under \textit{offload}, the controller uploads a fresh raw observation to the ground for remote inference. The SMDP holding time is \(L_O=U_{\mathrm{img}}+C_L\), where \(U_{\mathrm{img}}\) is the raw-data upload time and \(C_L\) is the ground inference time. This action is feasible when the raw upload can be completed within the current contact window, i.e., \(r(\phi)\ge U_{\mathrm{img}}\).
		\end{list}
	}
	
	Thus, the admissible action set for mode \(m=(\phi,q,\tau)\) is
	\begin{equation}
		\begin{aligned}
			\mathcal{A}(m)
			=&\{\textit{idle},\textit{compute}\}
			\cup \{\textit{tx}: q=1,\ r(\phi)\ge U_{\mathrm{tx}}\} \\
			&\cup \{\textit{offload}: r(\phi)\ge U_{\mathrm{img}}\}.
		\end{aligned}
	\end{equation}
	
	\subsubsection{State Transitions and Cost}
	For each action \(a\in\mathcal{A}_0\) with duration \(L_a\), define the contact-phase update operator:
	\begin{equation}
		\Phi_a(\phi)\triangleq (\phi+L_a)\bmod P.
	\end{equation}
	In particular, \(\Phi_I\), \(\Phi_C\), \(\Phi_T\), and \(\Phi_O\) correspond to the actions \textit{idle}, \textit{compute}, \textit{tx}, and \textit{offload}, respectively. The mode transition is
	\begingroup
	\small
	\begin{equation}
		\label{eq:mode-transition}
		f(m,a)=
		\begin{cases}
			\bigl(\Phi_I(\phi),\,1,\,\min\{\tau+L_I,\hat{\Delta}\}\bigr),
			& a=\textit{idle},\ q=1,\\
			\bigl(\Phi_I(\phi),\,0,\,0\bigr),
			& a=\textit{idle},\ q=0,\\
			\bigl(\Phi_C(\phi),\,1,\,\min\{L_C,\hat{\Delta}\}\bigr),
			& a=\textit{compute},\\
			\bigl(\Phi_T(\phi),\,0,\,0\bigr),
			& a=\textit{tx},\\
			\bigl(\Phi_O(\phi),\,q,\,q\min\{\tau+L_O,\hat{\Delta}\}\bigr),
			& a=\textit{offload}.
		\end{cases}
	\end{equation}
	\endgroup
	For tractability, semantic transmission is single-use: once selected, the cached object is consumed regardless of success. This makes the post-action mode transition deterministic\footnote{Allowing a failed semantic transmission to remain in the cache is a meaningful extension for future work.}.
	
	Let \(\xi_T(m)\) and \(\xi_O\) denote the reset values after successful semantic transmission and successful raw-data offloading, respectively. A successful semantic transmission delivers a cached result that is already \(\tau\) slots old and incurs an additional transmission delay \(L_T\), whereas a successful offloaded update corresponds to a freshly sensed observation and therefore depends only on the total offloading delay \(L_O\). Accordingly, $\xi_T(m)$ and $\xi_O$ are given by
	\begin{equation}
		\label{eq:success-reset}
		\begin{aligned}
			\xi_T(m) &= \min\{\tau+L_T,\hat{\Delta}\},\quad
			\xi_O = \min\{L_O,\hat{\Delta}\}.
		\end{aligned}
	\end{equation}
	If no task-valid update reaches the ground, \textsc{AoInf} ages during the action holding time. We denote this evolution by
	\begin{equation}
		\label{eq:failure-update}
		T_a(\Delta)=\min\{\Delta+L_a,\hat{\Delta}\}.
	\end{equation}
	Let \(p_a\) denote the probability that action \(a\) yields a task-valid ground update, namely \(p_a=0\) for \(a\in\{\textit{idle},\textit{compute}\}\), \(p_a=p_T\) for \(a=\textit{tx}\), and \(p_a=p_O\) for \(a=\textit{offload}\). Using the shorthand \(\xi_a(m_t)=\xi_T(m_t)\) for \(a=\textit{tx}\) and \(\xi_a(m_t)=\xi_O\) for \(a=\textit{offload}\), the \textsc{AoInf} transition can be written as
	\begin{equation}
		\label{eq:aoinf-transition}
		\Delta_{t+1}=
		\begin{cases}
			T_a(\Delta_t), & \text{with probability } 1-p_a,\\
			\xi_a(m_t), & \text{with probability } p_a.
		\end{cases}
	\end{equation}
	
	The one-step SMDP cost is the accumulated truncated \textsc{AoInf} area during the action execution interval:
	\begin{equation}
		R(\Delta,a)=\sum_{i=0}^{L_a-1}\min\{\Delta+i,\hat{\Delta}\}.
	\end{equation}
	
	\section{Optimality Equation and Solution}
	\subsection{Average-Cost SMDP} Since $\mathcal{S}$ and $\mathcal{A}(m)$ are finite after truncation, the proposed framework is a finite-state average-cost SMDP. Let $\rho^\star$ denote the optimal long-run average \textsc{AoInf} per slot defined in \eqref{rhostardefinition}, and let $V(\Delta,m)$ denote the corresponding differential value function. The average-cost optimality equation (ACOE) is given by~\cite[Eq. 4.1]{howard1960dynamic}:
	\begin{equation}\label{eq:acoe}
		V(\Delta,m)=\min_{a\in\mathcal{A}(m)} Q(\Delta,m,a),
	\end{equation}
	where the action-value term $Q(\Delta,m,a)$ is, with \(m^+=f(m,a)\):
	\begin{equation}\label{eq:q-value}
		\begin{aligned}
			Q(\Delta,m,a)=\; & R(\Delta,a)-\rho^\star L_a
			+ (1-p_a)V\bigl(T_a(\Delta),m^+\bigr) \\
			& + p_aV\bigl(\xi_a(m),m^+\bigr),
		\end{aligned}
	\end{equation}
	where \(f(m,a)\) is given in \eqref{eq:mode-transition}.
	
	Because actions contribute different numbers of physical slots,
	decision-epoch costs cannot be compared. We therefore
	use the SMDP-to-MDP construction~\cite{puterman1994mdp}. The transformation preserves stationary optimal policies and yields a uniform-step Bellman equation to which RVI can be applied.
	\subsection{SMDP-to-MDP Transformation}
	Let $s=(\Delta,m)$, and let $P(s'\mid s,a)$ denote the SMDP transition probability determined by \eqref{eq:mode-transition} and \eqref{eq:aoinf-transition}. Choose a constant $\theta\in(0,\min_a L_a]$. The transformed one-step cost and transition kernel are then defined as follows:
	\begin{equation}\label{newmdpcost}
		\bar{R}_{\theta}(s,a) \triangleq \theta \frac{R(\Delta,a)}{L_a}.
	\end{equation}
	
	\begin{equation}\label{newtransition}
		\bar{P}_\theta(s'\mid s,a)=
		\begin{cases}
			\dfrac{\theta}{L_a} P(s'\mid s,a), & s' \neq s,\\
			1-\dfrac{\theta}{L_a}+\dfrac{\theta}{L_a} P(s\mid s,a), & s'=s.
		\end{cases}
	\end{equation}
	
	\begin{theorem}[Equivalent MDP]
		The transformed MDP defined by \eqref{newmdpcost} and \eqref{newtransition} has the same stationary optimal policies as the original SMDP. If $\rho^{\star}$ is the optimal average \textsc{AoInf} of the SMDP, then the optimal average cost of the transformed MDP is $\bar{\rho}^{\star}=\theta\rho^{\star}$.
	\end{theorem}
	\begin{proof}
		Substituting \eqref{newmdpcost}--\eqref{newtransition} into the transformed MDP ACOE and setting \(\bar{\rho}^\star=\theta\rho^\star\) yields
		\begin{equation}
			\rho^\star
			=
			\min_{a\in\mathcal{A}(m)}
			\frac{
				R(\Delta,a)
				+
				\sum_{s'}P(s'\mid s,a)V(s')
				-
				V(s)
			}{
				L_a
			},
			\label{eq:smdp-ratio-form}
		\end{equation}
		which is equivalent to the SMDP optimality equation
		\eqref{eq:acoe} and \eqref{eq:q-value}. The factor \(\theta\) is positive and action-independent after the division by \(L_a\), so the minimizing actions are unchanged. Hence, stationary optimal policies are preserved.
	\end{proof}
	
	\subsection{Structural Results}
	
	We record two structural properties that explain the policy patterns observed in the numerical results. The first is monotonicity in the current ground-side \textsc{AoInf}; the second is a cache-age threshold for choosing between cached semantic transmission and onboard recomputation.
	
	\begin{lemma}[Monotonicity in the current \textsc{AoInf}]
		\label{lem:delta-monotone}
		For every fixed mode \(m\), the differential value function of the transformed MDP can be chosen nondecreasing in the ground \textsc{AoInf}:
		\[
		\bar V(\Delta_1,m)\le \bar V(\Delta_2,m),
		\qquad
		1\le \Delta_1\le \Delta_2\le \hat{\Delta}.
		\]
	\end{lemma}
	
	\begin{proof}
		Initialize value iteration with a function nondecreasing in \(\Delta\). For any fixed action \(a\), the transformed one-step cost is nondecreasing in \(\Delta\), since
		\[
		R(\Delta,a)=\sum_{i=0}^{L_a-1}\min\{\Delta+i,\hat{\Delta}\}
		\]
		is nondecreasing. The no-reset update
		\[
		T_a(\Delta)=\min\{\Delta+L_a,\hat{\Delta}\}
		\]
		is also nondecreasing, while successful-update ages are independent of, or no worse than, the current \(\Delta\). Hence each action-value term preserves monotonicity in \(\Delta\), and so does the minimum over feasible actions. The claim follows by induction over the normalized value-iteration sequence and the finite-state limiting argument.
	\end{proof}
	
	This monotonicity implies that larger \textsc{AoInf} cannot reduce the future cost-to-go. We use it to characterize the local competition between transmitting a cached semantic summary and replacing it by a fresh onboard computation.
	
	\begin{proposition}[Cache-age threshold between transmission and recomputation]
		\label{prop:tx-compute-threshold}
		Fix \(\Delta\) and \(\phi\), and suppose the cache is nonempty and semantic transmission is feasible. Define
		\[
		D_{T,C}(\tau)
		\triangleq
		\bar Q(\Delta,\phi,1,\tau,\textit{tx})
		-
		\bar Q(\Delta,\phi,1,\tau,\textit{compute}).
		\]
		Then \(D_{T,C}(\tau)\) is nondecreasing in \(\tau\). Consequently, the set
		\[
		\{\tau: D_{T,C}(\tau)\le 0\}
		\]
		is a prefix interval of the cache-age axis. Equivalently, there exists a possibly degenerate threshold
		\[
		\tau_{T/C}^\star(\Delta,\phi)\in\{-1,0,\ldots,\hat{\Delta}\}
		\]
		such that cached semantic transmission is preferred to recomputation only when
		\[
		\tau\le \tau_{T/C}^\star(\Delta,\phi).
		\]
	\end{proposition}
	
	\begin{proof}
		The compute action overwrites the cache and therefore its action value is independent of the current cache age \(\tau\). By contrast, under cached semantic transmission, the successful post-action age is nondecreasing in \(\tau\). By Lemma~\ref{lem:delta-monotone}, the corresponding continuation value is nondecreasing in \(\tau\). Hence \(\bar Q(\Delta,\phi,1,\tau,\textit{tx})\) is nondecreasing in \(\tau\), while \(\bar Q(\Delta,\phi,1,\tau,\textit{compute})\) is independent of \(\tau\). Thus \(D_{T,C}(\tau)\) is nondecreasing. On the finite cache-age grid, its sublevel set is a prefix interval, which proves the threshold claim.
	\end{proof}
	
	The proposition implies that cached semantics are useful only below a mode-dependent age threshold. Once the cached result becomes too stale, the controller prefers recomputation or another feasible refresh action.
	
	\subsection{Solution Method}
	
	We solve the transformed MDP by normalized RVI, a standard method for finite average-cost MDPs~\cite{howard1960dynamic,puterman1994mdp}, to optimize the transformed objective \(\bar{\rho}^\star=\theta\rho^\star\).
	
	Let \(s=(\Delta,m)\) and \(s^\dagger\) be a reference state. Given a relative value iterate \(\bar V_k\), the dynamic-programming backup is
	\begin{equation}
		\widetilde V_{k+1}(s)
		=
		\min_{a\in\mathcal{A}(m)}
		\left\{
		\bar R_\theta(s,a)
		+
		\sum_{s'\in\mathcal{S}}
		\bar P_\theta(s'\mid s,a)\bar V_k(s')
		\right\}.
		\label{eq:rvi-backup}
	\end{equation}
	The next iterate is normalized as
	\begin{equation}
		\bar V_{k+1}(s)
		=
		\widetilde V_{k+1}(s)
		-
		\widetilde V_{k+1}(s^\dagger),
		\label{eq:rvi-normalization}
	\end{equation}
	which removes the additive indeterminacy without changing greedy
	action comparisons. Iteration stops when $\operatorname{span}(x)=\max_{s\in\mathcal{S}}x(s)-\min_{s\in\mathcal{S}}x(s)$ and terminates when    \\$\operatorname{span}\!\left(\bar V_{k+1}-\bar V_k\right)
	\le \varepsilon_{\mathrm{RVI}}$. The details of the algorithm are shown in Algorithm \ref{alg:normalized-rvi}.
	
	\begin{algorithm}[t]
		\caption{Normalized RVI for the Transformed MDP}
		\label{alg:normalized-rvi}
		\small
		\begin{algorithmic}[1]
			\Require State space \(\mathcal{S}\), feasible action sets \(\mathcal{A}(m)\), transformed cost \(\bar R_\theta\), transformed kernel \(\bar P_\theta\), reference state \(s^\dagger\), tolerance \(\varepsilon_{\mathrm{RVI}}\)
			\State Initialize \(\bar V_0(s)=0\) for all \(s\in\mathcal{S}\), and set \(k\gets 0\)
			\Repeat
			\For{each state \(s=(\Delta,m)\in\mathcal{S}\)}
			\State \(\widetilde V_{k+1}(s)\gets \min_{a\in\mathcal{A}(m)}\{\bar R_\theta(s,a) + \sum_{s'\in\mathcal{S}} \bar P_\theta(s'\mid s,a)\bar V_k(s')\}\)
			\EndFor
			\State \(\bar V_{k+1}(s)\gets \widetilde V_{k+1}(s)-\widetilde V_{k+1}(s^\dagger)\), \(\forall s\in\mathcal{S}\)
			\State \(d_k\gets\operatorname{span}(\bar V_{k+1}-\bar V_k)\), \(k\gets k+1\)
			\Until{\(d_k\le \varepsilon_{\mathrm{RVI}}\)}
			\State \textbf{return} \(\hat\pi(s)\)
			\Statex \(\in \arg\min_{a\in\mathcal{A}(m)}\{\bar R_\theta(s,a) + \sum_{s'\in\mathcal{S}} \bar P_\theta(s'\mid s,a)\bar V_k(s')\}\)
		\end{algorithmic}
	\end{algorithm}
	
	\section{Numerical Results}
	\begin{table}[t]
		\centering
		\caption{Baseline numerical parameters.}
		\label{tab:baseline-parameters}
		\begin{tabular}{lc}
			\toprule
			\textbf{Parameter} & \textbf{Value} \\
			\midrule
			\textsc{AoInf} cap \(\hat{\Delta}\) & \(40\) slots \\
			Visibility period \(P\) & \(30\) slots \\
			Visible-window length \(W\) & \(20\) slots \\
			Onboard compute duration \(C_S\) & \(2\) slots \\
			Semantic transmission \(U_{\mathrm{tx}}\) & \(3\) slots \\
			Raw-data upload \(U_{\mathrm{img}}\) & \(5\) slots \\
			Ground inference duration \(C_L\) & \(1\) slot \\
			\bottomrule
		\end{tabular}
	\end{table}
	Unless otherwise stated, all results use the baseline parameter setting in Table~\ref{tab:baseline-parameters}. In practice, LEO contact ratios are often around 5\%--15\%; here we use larger visible-window settings (e.g., \(W=20\) with \(P=30\)) to make the policy behavior easier to visualize.
	
	\subsection{Representative Trajectory}
	
	% \begin{figure}[t]
		%   \centering
		%   \includegraphics[width=0.5\textwidth]{5_5sim2v3_3.pdf}
		%   \caption{Trajectories under the optimized hybrid policy. (a) }
		%   \Description{A three-panel steady-state simulation figure spanning both columns. Panel (a) shows the ground Age of Inference over time with repeated sawtooth growth and resets. Panel (b) shows the selected action sequence with visibility windows shaded. Panel (c) shows cached semantic age over time.}
		%   \label{fig:trajectory}
		% \end{figure}
	
	\begin{figure}[t]
		\centering
		\includegraphics[width=0.98\columnwidth]{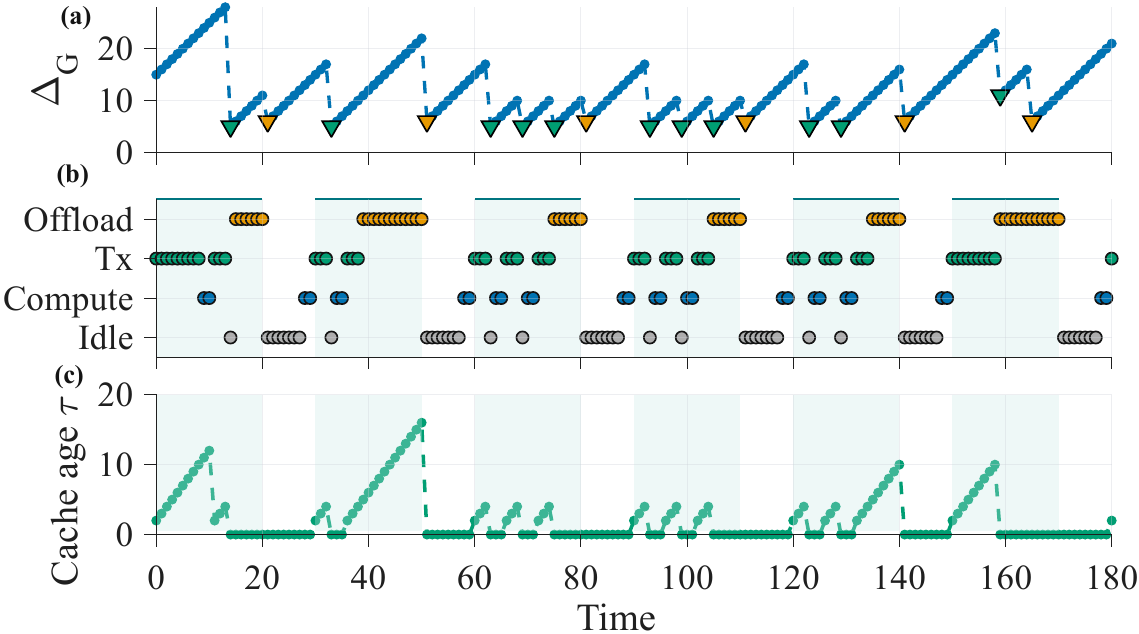}
		\caption{Trajectories under the optimized hybrid policy with $p_S=0.6$ and $p_L=0.7$. Panel (a) shows the ground Age of Inference over time with repeated sawtooth growth and resets, (b) shows the selected action sequence with visibility windows shaded, (c) shows cached semantic age over time.}
		\Description{A three-panel steady-state simulation figure spanning both columns. Panel (a) shows the ground Age of Inference over time with repeated sawtooth growth and resets. Panel (b) shows the selected action sequence with visibility windows shaded. Panel (c) shows cached semantic age over time.}
		\label{fig:trajectory}
	\end{figure}
	
	Figure~\ref{fig:trajectory} illustrates how the optimized controller coordinates computation, caching, and intermittent connectivity.
	
	In the top panel, \textsc{AoInf} exhibits sawtooth-like growth between successful updates. Unlike classical AoI, reset levels are not uniform. Successful semantic transmission reduces \textsc{AoInf} to cached semantic age plus transmission duration, whereas successful offloading resets it to the total offloading delay. As a result, the figure has multiple reset levels rather than a single baseline.
	
	Communication occurs within visible windows, reflecting the feasibility constraint imposed by $r(\phi_t)$. Within these windows, the controller switches between transmission and offloading based on cache state and expected freshness gain.
	
	Compute actions fill the cache, while successful transmissions deplete it. Cached semantics typically remain fresh, so stale reuse is rarely preferred.
	
	Overall, the policy coordinates computation and communication instead of acting greedily.
	
	The trajectory also reveals a threshold structure in the optimal policy. In blackout phases, the choice reduces to waiting or local recomputation. In visible phases with an empty cache, the controller balances waiting, recomputation, and raw-data offloading. With a nonempty cache, semantic transmission is preferred only when the content is sufficiently fresh. For fixed mode variables, larger \textsc{AoInf} values favor earlier updates.
	
	\begin{figure}[t]
		\centering
		\includegraphics[trim=20 220 18 180,clip,width=0.98\columnwidth]{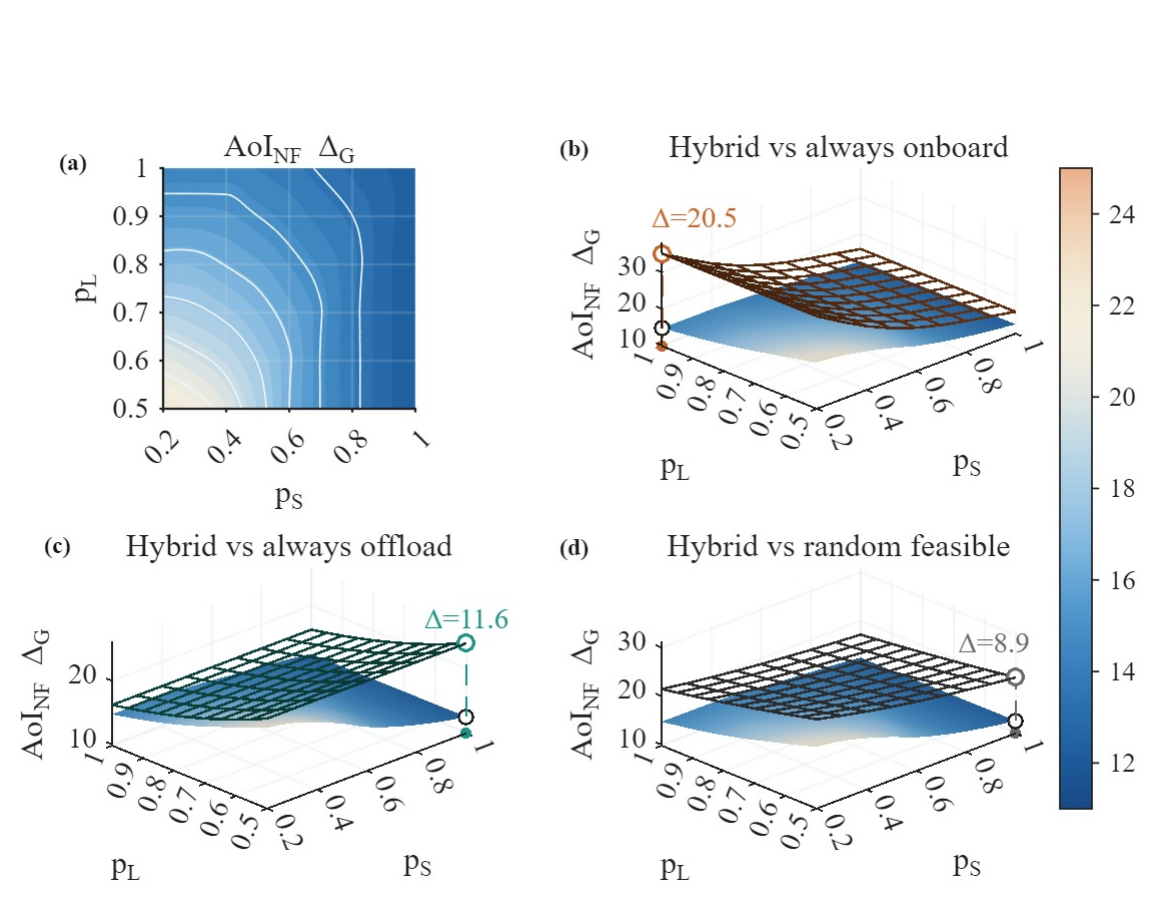}
		\caption{Average \textsc{AoInf} comparison of the optimized hybrid policy against three benchmark policies across semantic-success and offloading-success regimes.}
		\Description{A three-dimensional comparison plot showing the optimized hybrid policy and three benchmark policies over different semantic-success and offloading-success probability settings, with average Age of Inference on the vertical axis.}
		\label{fig:benchmark-comparison}
	\end{figure}
	
	\subsection{Benchmark Policies}
	We compare against Random, Onboard, and Offload.
	
	Figure~\ref{fig:benchmark-comparison} compares the optimized hybrid policy with these baselines across semantic-success and offloading-success regimes. Random chooses feasible actions uniformly, Onboard prioritizes local computation with cached transmission, and Offload uses raw-data offloading when feasible.
	
	The optimized hybrid policy achieves lower average \textsc{AoInf} than all policies across most regimes, demonstrating the benefit of jointly coordinating onboard computation, semantic caching, and raw-data offloading.
	
	The advantage of semantic transmission increases as the semantic-success probability \(p_S\) increases. In this regime, cached semantic updates are more likely to succeed, making low-delay semantic delivery particularly effective.
	
	When the offloading success probability \(p_L\) is high, the offload-only policy improves because fresh raw observations can reliably reach the ground. Still, the hybrid controller retains an advantage by adapting to the current system state.
	
	The random policy performs poorly because it ignores visibility and freshness. These results show that the proposed hybrid controller adapts to heterogeneous communication and inference conditions.
	
	\section{Conclusion}
	This paper develops an average-cost control framework for hybrid semantic caching and offloading under inference freshness in LEO networks. Using \textsc{AoInf} as the performance metric, we formulate the problem as an average-cost SMDP and obtain the optimal policy through an equivalent MDP formulation. Numerical results show that the hybrid policy adapts to communication availability and inference-success reliability and outperforms single-mechanism baselines across representative regimes.
	
	\newpage
	\section*{Acknowledgment}
	This work was supported by Türk Telekomünikasyon A.Ş. through the project ``Goal-Oriented Scheduling for Networked Control,'' Project No. 136124, and by the European Union through the ERC Advanced Grant GO SPACE, Grant No. 101122990. Views and opinions expressed are, however, those of the authors only and do not necessarily reflect those of the European Union or the European Research Council Executive Agency. Neither the European Union nor the granting authority can be held responsible for them.
	
	\begingroup
	\footnotesize
	\setlength{\bibsep}{0pt}
	\setlength{\itemsep}{0pt}
	\bibliographystyle{IEEEtranN}
	\bibliography{references}
	\endgroup
	
\end{document}